\documentclass[runningheads]{llncs}
\usepackage{amsmath}

\usepackage{cite}
\usepackage{graphicx}%%
\usepackage{array}
\usepackage{url}
\usepackage{times}
\usepackage{amssymb}

\usepackage{latexsym}
\usepackage{graphics}
\usepackage{url}
\usepackage{verbatim}
\usepackage{color}
\usepackage{lineno}
\usepackage{multicol}
\usepackage{caption}
\usepackage{subcaption}
\usepackage{lineno}
\usepackage{soul}
\usepackage{comment}
\usepackage{datenumber}

\usepackage{paralist}
\usepackage{enumerate}
\usepackage{thm-restate}% CTAN:macros/experimental/thmtools
\usepackage[noend]{algpseudocode}
\usepackage{algorithm}

\usepackage{listings}
\usepackage{xcolor}
\usepackage[misc]{ifsym}

\definecolor{codegreen}{rgb}{0,0.6,0}
\definecolor{codegray}{rgb}{0.5,0.5,0.5}
\definecolor{codepurple}{rgb}{0.58,0,0.82}
\definecolor{backcolour}{rgb}{0.95,0.95,0.92}

\algdef{SE}[DOWHILE]{Do}{doWhile}{\algorithmicdo}[1]{\algorithmicwhile\ #1}%

\lstdefinestyle{mystyle}{
  backgroundcolor=\color{backcolour},   commentstyle=\color{codegreen},
  keywordstyle=\color{magenta},
  numberstyle=\tiny\color{codegray},
  stringstyle=\color{codepurple},
  basicstyle=\ttfamily\footnotesize,
  breakatwhitespace=false,         
  breaklines=true,                 
  captionpos=b,                    
  keepspaces=true,                 
  numbers=left,                    
  numbersep=5pt,                  
  showspaces=false,                
  showstringspaces=false,
  showtabs=false,                  
  tabsize=2
}

\usepackage{fancybox}
\usepackage{tikz}
\usepackage[T1]{fontenc}
\usepackage{fourier}

\usepackage[english]{babel}															% English language/hyphenation
\usepackage[protrusion=true,expansion=true]{microtype}	
\usepackage{amsfonts} % Math packages

\usepackage{svg}

\usepackage{fancyhdr}
\numberwithin{equation}{section}		% Equationnumbering: section.eq#
\numberwithin{figure}{section}			% Figurenumbering: section.fig#
\numberwithin{table}{section}				% Tablenumbering: section.tab#

\begin{document}
\title{Recursive and iterative approaches to generate rotation Gray codes for stamp foldings and semi-meanders}
\author{Bowie~Liu \and Dennis~Wong \and Chan-Tong Lam \and Marcus Im}%^{\textrm{~\Letter}}$}
\authorrunning{B. Liu and D. Wong}
\institute{Macao Polytechnic University, Macao, China \\
\email{\{bowen.liu, cwong, ctlam, marcusim\}@mpu.edu.mo}}
\maketitle              % typeset the header of the contribution
\begin{abstract}
% We first present a simple recursive algorithm that generates cyclic rotation Gray codes for stamp foldings and semi-meanders, where consecutive strings differ by a stamp rotation. 
We first present a simple recursive algorithm that generates cyclic rotation Gray codes for stamp foldings and semi-meanders, where consecutive strings differ by a stamp rotation. 
These are the first known Gray codes for stamp foldings and semi-meanders, and we thus solve an open problem posted by Sawada and Li in [Electron. J. Comb. 19(2), 2012]. 
We then introduce an iterative algorithm that generates the same rotation Gray codes for stamp foldings and semi-meanders.
Both the recursive and iterative algorithms generate stamp foldings and semi-meanders in constant amortized time and $O(n)$-amortized time per string respectively, using a linear amount of memory.

\keywords{Stamp foldings  \and Meanders \and Semi-meanders \and Reflectable language \and Binary reflected Gray code \and Gray code \and CAT algorithm.}
\end{abstract}
\section{Introduction}
~\label{sec:intro}

A \emph{stamp folding} is a way to fold a linear strip of $n$ stamps into a single pile, with the assumption that the perforations between the stamps are infinitely elastic. 
As an example, Figure~\ref{fig:stamp} illustrates the sixteen stamp foldings for $n = 4$. 
We always orient a pile of stamps horizontally, with the stamps facing left and right and the perforations facing up and down. 
In addition, the perforation between stamp $1$ and stamp $2$ is located at the bottom.
The sixteen piles of stamps for $n=4$ can thus be obtained by contracting the horizontal lines in Figure~\ref{fig:stamp}.
Each stamp folding can be represented by a unique permutation $(\pi(1) \pi(2) \cdots \pi(n))$, where $\pi(i)$ is the stamp at position $i$ in the pile when considering it from left to right.
The permutations that correspond to the sixteen stamp foldings in Figure~\ref{fig:stamp} for $n=4$ are as follows: 
\begin{center}
%1234, 1243, 1342, 1432, 2341, 2431, 3214, 3421, 4123, 4321. 
1234, 1243, 1342, 1432, 2134, 2143, 2341, 2431, \\
3124, 3214, 3412, 3421, 4123, 4213, 4312, 4321. 
\end{center}
In the rest of this paper, we use this permutation representation to represent stamp foldings. 
An alternative permutation representation can be found in~\cite{HOFFMANN1986170}.
Note that not all permutations correspond to a valid stamp folding. 
For example, the permutation $1423$ requires the strip of stamps to intersect itself and is not a valid stamp folding.

\begin{figure}[t]
\begin{center}
\includegraphics[width=1\columnwidth]{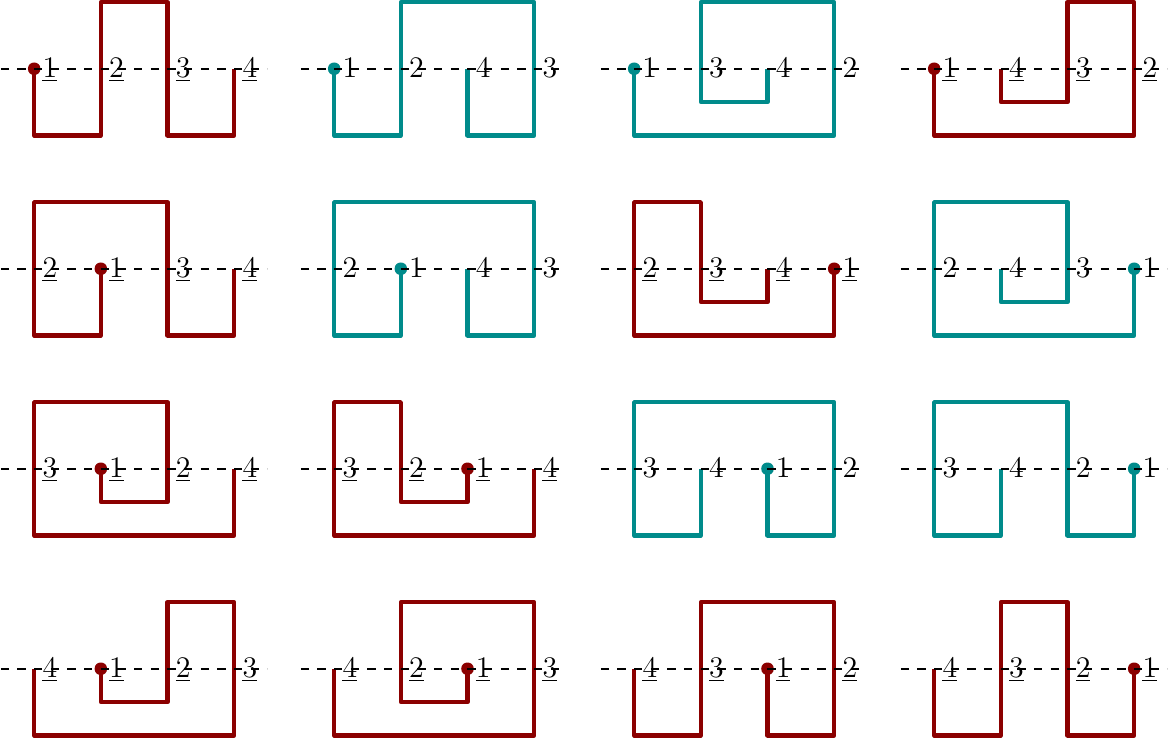}
\end{center}
\caption{Stamp foldings of order four. The stamp with a dot is stamp $1$. Stamp foldings that are underlined (in red color) are semi-meanders. }
\label{fig:stamp}
\end{figure}

A \emph{semi-meander} is a stamp folding with the restriction that stamp $n$ can be seen from above when $n$ is even, and stamp $n$ can be seen from below when $n$ is odd.
In other words, stamp $n$ is not obstructed by perforations between other stamps and remains visible from either above or below, depending on its parity.
For example, $1234$ is a semi-meander, while $2143$ is not a semi-meander since stamp $4$ is blocked by the perforation between stamp $2$ and stamp $3$ and cannot be seen from above. 
The permutations that correspond to the ten  semi-meanders in Figure~\ref{fig:stamp} (underlined and in red color) for $n=4$ are given as follows: 
\begin{center}
%1234, 1243, 1342, 1432, 2341, 2431, 3214, 3421, 4123, 4321. 
1234, 1432, 2134, 2341, 3124, 3214, 4123, 4213, 4312, 4321. 
\end{center}

The study of stamp foldings and related combinatorial objects has a long history and has traditionally attracted considerable interest from mathematicians~\cite{ 10.1145/1721837.1721858, combos, france_van_der_poorten_1981, Iwan_Jensen_2000, KOEHLER1968135, DBLP:journals/ajc/Legendre14, lucas_number, Lunnon1968AMP, lague_stamp, DBLP:journals/combinatorics/SawadaL12, doi:10.1137/1016111, touchard_1950}. 
For example, the five foldings of four blank stamps appear at the front cover of the book \emph{A Handbook of Integer Sequences} by Sloane~\cite{doi:10.1137/1016111}, which is the ancestor of the well-known \emph{Online Encyclopedia of Integer Sequences}~\cite{Sloane2010OESI}. 
Lucas~\cite{lucas_number} first posed the enumeration problem for stamp foldings in 1891, questioning the number of ways a strip of $n$ stamps could be folded. % by asking in how many ways a strip of $n$ stamps can be folded.
See~\cite{DBLP:journals/ajc/Legendre14} for a brief history of the development of the enumeration problem of stamp foldings.
Stamp foldings and related combinatorial objects have lots of applications, ranging from robot coverage path planning~\cite{9561433} and conditioned random walks~\cite{Iordache2017}, to protein folding~\cite{schweitzer2012protein}. 

The enumeration sequences for stamp foldings and semi-meanders are A000136 and A000682 in the Online Encyclopedia of Integer Sequences, respectively~\cite{Sloane2010OESI}.
The first ten terms for the enumeration sequences of stamp foldings and semi-meanders are as follows: 
\begin{itemize}
\item Stamp foldings: 1, 2, 6, 16, 50, 144, 462, 1392, 4536, and 14060;
\item Semi-meanders: 1, 2, 4, 10, 24, 66, 174, 504, 1406, and 4210. 
\end{itemize}
Although a large number of terms of the folding sequences have been computed, no closed formula has been found for both enumeration sequences.

One of the most important aspects of combinatorial generation is to list the instances of a combinatorial object so that consecutive instances differ by a specified \emph{closeness condition} involving a constant amount of change. 
Lists of this type are called \emph{Gray codes}. 
This terminology is due to the eponymous \emph{binary reflected Gray code} (BRGC) by Frank Gray, which orders the $2^n$ binary strings of length $n$ so that consecutive strings differ by one bit. 
For example, when $n = 4$ the order is
\begin{center} 
0000, 1000, 1100, 0100, 0110, 1110, 1010, 0010, \\ 0011, 1011, 1111, 0111, 0101, 1101, 1001, 0001.
\end{center}
We note that the order above is \emph{cyclic} because the last and first strings also differ by the closeness condition, and this property holds for all $n$.
The BRGC listing is a 1-Gray code in which consecutive strings differ by one symbol change. 
In this paper, we focus on \emph{rotation Gray code} in which consecutive strings differ by a stamp rotation.

An interesting related problem is thus to discover Gray codes for stamp foldings and semi-meanders.
There are several algorithms to generate stamp foldings and semi-meanders. 
Lunnon~\cite{Lunnon1968AMP} in 1968 provided a backtracking algorithm that exhaustively generates stamp foldings after considering rotational equivalence and equivalence of the content between the first crease and the second crease.
More recently, Sawada and Li~\cite{DBLP:journals/combinatorics/SawadaL12} provided an efficient algorithm that exhaustively generates stamp foldings and semi-meanders in constant amortized time per string. 
However, the listings produced by these algorithms are not Gray codes. 
The problem of finding Gray codes for stamp foldings and semi-meanders is listed as an open problem in the paper by Sawada and Li~\cite{DBLP:journals/combinatorics/SawadaL12}, and also in the Gray code survey by M\"{u}tze~\cite{mutze2023combinatorial}. 
In this paper, we solve this open problem by providing the first known {rotation Gray codes} for stamp foldings and semi-meanders, where consecutive strings differ by a stamp rotation. 
The formal definition of stamp rotation is provided in Section~\ref{sec:algorithm}. 
The algorithms generate stamp foldings in constant amortized time per string, and semi-meanders in $O(n)$-amortized time per string {respectively}, 
using a linear amount of memory. 

A preliminary version of this paper appeared at the IWOCA 2023 conference~\cite{10.1007/978-3-031-34347-6_23}.

The rest of the paper is outlined as follows.
In Section~\ref{sec:algorithm}, we describe a simple recursive algorithm to generate cyclic rotation Gray codes for stamp foldings and semi-meanders.
Then in Section~\ref{sec:proves}, we prove the Gray code property and show that the recursive algorithm generates each stamp folding and semi-meander in constant amortized time and $O(n)$-amortized time respectively, using a linear amount of memory.
In Section~\ref{sec:iterative}, we describe a simple iterative algorithm which generates the same cyclic rotation Gray codes for stamp foldings and semi-meanders as the recursive algorithm. 
We also prove that the iterative algorithm generates each stamp folding and semi-meander in constant amortized time and $O(n)$ time respectively. 
Finally, we conclude our paper in Section~\ref{sec:remarks}.

\section{A recursive algorithm to generate Gray codes for semi-meanders and stamp foldings}
\label{sec:algorithm}
In this section, we first describe a simple recursive algorithm to generate a cyclic rotation Gray code for semi-meanders. 
We then describe simple modifications to the algorithm to generate a cyclic rotation Gray code for stamp foldings. 

\begin{figure}[t]
\begin{center}
\includegraphics[width=1\columnwidth]{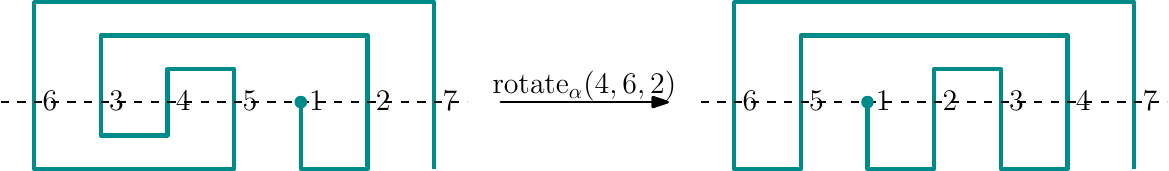}
\end{center}
\caption{Stamp rotation. The stamp folding $6512347$ can be obtained by applying a stamp rotation ${\textrm{rotate}}_\alpha(4, 6, 2)$ on $\alpha = 6345127$.}
\label{fig:mrep_4bits}
\end{figure}

Consider a stamp folding $\alpha$ = $p_1 p_2 \cdots p_n$.
A \emph{stamp rotation} of the $i$-th to $j$-th symbols of a stamp folding $\alpha$ into its $k$-th position with $k < i \leq j \leq n$, denoted by ${\textrm{rotate}}_\alpha(i, j, k)$, is the string
$p_1 p_2 \cdots p_{k-1} p_{i} p_{i+1} \cdots p_j p_k p_{k+1} \cdots p_{i-1} p_{j+1} p_{j+2} \cdots p_n$. 
As an example, given a stamp folding (semi-meander) $\alpha = 6345127$, 
Figure~\ref{fig:mrep_4bits} illustrates the stamp folding (semi-meander) $6512347$ obtained by applying ${\textrm{rotate}}_\alpha(4, 6, 2)$ on $\alpha$ (illustrated as $6\overleftarrow{34\underline{512}}7 \xrightarrow{} 6512347$, or equivalently $6\overrightarrow{\underline{34}{512}}7 \xrightarrow{} 6512347$).
Note that not all stamp rotations with $k < i \leq j \leq n$ can generate a valid stamp folding.
To simplify the notation, we define two special stamp rotations. 
A \emph{left rotation} of the suffix starting at position $i$ of a stamp folding $\alpha$, denoted by $\overleftarrow{\textrm{rotate}}_\alpha(i)$, is the string obtained by rotating the suffix $p_i p_{i+1} \cdots p_n$ to the front of $\alpha$, that is $\overleftarrow{\textrm{rotate}}_\alpha(i) = {\textrm{rotate}}_\alpha(i, n, 1) = p_i p_{i+1} \cdots p_n p_1 p_2 \cdots p_{i-1}$.
Similarly, a \emph{right rotation} of the prefix ending at position $j$ of a stamp folding $\alpha$, denoted by $\overrightarrow{\textrm{rotate}}_\alpha(j)$, is the string obtained by rotating the prefix $p_1 p_2 \cdots p_j$ to the end of $\alpha$, that is $\overrightarrow{\textrm{rotate}}_\alpha(j) = {\textrm{rotate}}_\alpha(j+1, n, 1) = p_{j+1} p_{j+2} \cdots p_n p_1 p_2 \cdots p_j$. 
When the context is clear, we use $\overleftarrow{\textrm{rotate}}(i)$ and $\overrightarrow{\textrm{rotate}}(j)$ to indicate $\overleftarrow{\textrm{rotate}}_\alpha(i)$ and $\overrightarrow{\textrm{rotate}}_\alpha(j)$ respectively. 
Observe that $\overleftarrow{\textrm{rotate}}_\alpha(t+1) = \overrightarrow{\textrm{rotate}}_\alpha(t)$. 
We also use the terms \emph{simple left rotation} and \emph{simple right rotation} to refer to $\overleftarrow{\textrm{rotate}}(n)$ and $\overrightarrow{\textrm{rotate}}(1)$ respectively.  

A \emph{string rotation} of a string $\alpha = p_1 p_2 \cdots p_n$ is the string $p_2 p_3 \cdots p_n p_1$ which is obtained by taking the first character $p_1$ of $\alpha$ and placing it in the last position.
The set of stamp foldings that are equivalent to a stamp folding $\alpha$ under string rotation is denoted by $\mathrm{Rots}(\alpha)$. %, and the set of all length $n$ necklaces is denoted by $\mathbf{N}(n)$.
For example, $\mathrm{Rots}(1243) = \{1243, 2431, 4312, 3124\}$.
Note that $|\mathrm{Rots}(\alpha)| = n$. 
The strings in $\mathrm{Rots}(\alpha)$ can be obtained by applying left rotation $\overleftarrow{\textrm{rotate}}_\alpha(i)$ on $\alpha$ for all integers $1 \leq i \leq n$, or applying right rotation $\overrightarrow{\textrm{rotate}}_\alpha(j)$ on $\alpha$ for all integers $1 \leq j \leq n$. 
We also define $I(e, \alpha)$ as the index of an element $e$ within a string $\alpha$. 
For example, $I(p_2, \alpha) = 2$ and $I(5, 6512347) = 2$.

\begin{lemma}~\cite{lague_stamp} \label{lem:rotation}
If $\alpha = p_1 p_2 \cdots p_n$ is a stamp folding, then $\beta \in \mathrm{Rots}(\alpha)$ is also a stamp folding.
\end{lemma}

Lemma~\ref{lem:rotation} implies that the set of stamp foldings is partitioned into equivalence classes under string rotation. 
Also, note that this property does not necessarily hold for semi-meanders. For example, the string $3124$ is a semi-meander but $1243 \in \mathrm{Rots}(3124)$ is not a semi-meander. 

The following lemmas follow from the definition of semi-meander. 
% proof of sub-meander is a meander?

\begin{comment}
\begin{figure}[t]
\begin{center}
\includegraphics[width=0.8\columnwidth]{mrep_4bits.pdf}
\end{center}
\caption{Representative necklaces of 4 bits stamp foldings}
\label{fig:mrep_4bits}
\end{figure}
\end{comment}

\begin{lemma} \label{lem:submeander}
The string $\alpha = p_1 p_2 \cdots p_n$ with $p_1 = n$ is a stamp folding of order $n$ if and only if $p_2 p_3 \cdots p_n$ is a semi-meander of order $n-1$. 

\end{lemma}

\begin{proof}
The backward direction is trivial. 
For the forward direction, assume by contrapositive that $p_2 p_3 \cdots p_n$ is not a semi-meander. 
If $p_2 p_3 \cdots p_n$ is not even a stamp folding of order $n-1$, then clearly $\alpha$ is not a stamp folding. 
Otherwise if $p_2 p_3 \cdots p_n$ is a stamp folding but not a semi-meander, now suppose $p_t = n-1$ and $2 \leq t \leq n$. 
Observe that by the definition of semi-meander, $p_t$ is blocked by a perforation between some stamps and thus it cannot connect to $p_1 = n$ without crossing any stamp. 
Thus $\alpha$ is also not a stamp folding. 
\end{proof}

\begin{corollary} \label{cor:rotationstamp}
    If $\alpha = p_1 p_2 \cdots p_n$ is a semi-meander with $p_1 = n$, then $p_2 p_3 \cdots p_n p_1$ is also a semi-meander.
\end{corollary}

\begin{proof}
By Lemma~\ref{lem:submeander}, $\alpha$ is a stamp folding and thus $p_2 p_3 \cdots p_n$ is a semi-meander of order $n-1$. 
Thus stamp $n-1$ of $p_2 p_3 \cdots p_n$ is not blocked by any perforation between stamps and can connect to a stamp at position $1$ or at position $n$ to produce the string $\alpha$ or $p_2 p_3 \cdots p_n p_1$.
Lastly, stamps at position $1$ and position $n$ are at the boundary and cannot be blocked by any perforation between stamps, thus $p_2 p_3 \cdots p_n p_1$ is a semi-meander. 
\end{proof}

\begin{lemma} \label{lem:rotatemeander}
Suppose $\alpha = p_1 p_2 \cdots p_n$ is a semi-meander where $p_n \ne n$ and $\beta = \overleftarrow{\textrm{rotate}}_\alpha(n-j+1)$ with $j \geq 1$ being the smallest possible integer such that $\beta$ is a semi-meander, then
\begin{itemize}
%\item $k = 1$ if $p_n = n$;
\item $j=1$ if $p_n = 1$ and $n$ is even;
\item $j = n - I(p_n + 1, \alpha) + 1$ if $p_n$ and $n$ have the same parity;
\item $j = n - I(p_n - 1, \alpha) + 1$ if $p_n$ and $n$ have different parities.
\end{itemize}
\end{lemma}

\begin{proof}
If $p_n = 1$ and $n$ is even, then clearly $j = 1$ since $p_n$ only connects to $p_i = 2$ for some $i < n$ and the perforation between $p_n = 1$ and $p_i = 2$ is at the bottom while stamp $n$ is extending in the upward direction.  % does not block stamp $n$ as $n$ is even. 
Otherwise, if $p_n$ and $n$ have the same parity, then assume W.L.O.G. that  the perforation between $p_n$ and $p_n+1$ is at the bottom. 
Clearly, $p_t = n$ is also extending in the downward direction. 
Since $\alpha$ is a semi-meander, $p_t = n$ can be seen from below and thus $t < I(p_n + 1, \alpha) < n$.
Now consider the string $\beta = \overleftarrow{\textrm{rotate}}_\alpha(n-j+1)$. 
When $1 < j < I(p_n + 1, \alpha)$, the perforation between $p_n + 1$ and $p_n$ of $\beta$ is at the bottom and would block stamp $n$ making $\beta$ not a semi-meander, and thus $j \geq I(p_n + 1, \alpha)$.  
Furthermore, observe that there is no perforation between $p_i$ and $p_k$ at the bottom with $I(p_n + 1, \alpha)<i$ and $k<I(p_n + 1, \alpha)$ as otherwise the perforation intersects with the perforation between $p_n + 1$ and $p_n$. 
Thus, $\beta = \overleftarrow{\textrm{rotate}}_\alpha(n-j+1)$ is a semi-meander when $j = I(p_n + 1, \alpha)$. 
The proof is similar when $p_n$ and $n$ have different parities.
\end{proof}

\newpage
\begin{lemma} \label{lem:rotatemeander2}
Suppose $\alpha = p_1 p_2 \cdots p_n$ is a semi-meander where $p_1 \ne n$ and 
% $\gamma = {RR}^k(\alpha)$ be the semi-meander obtained by right rotating $\alpha$ $k$ times such that $k$ is the smallest integer with $1 \leq k < n$, 
$\beta = \overrightarrow{\textrm{rotate}}_\alpha(j)$ with $j \geq 1$ being the smallest possible integer such that $\beta$ is a semi-meander,
then
\begin{itemize}
%\item $k = 1$ if $p_n = n$;
\item $j=1$ if $p_1 = 1$ and $n$ is even;
\item $j = I(p_1 + 1, \alpha)$ if $p_1$ and $n$ have the same parity;
\item $j = I(p_1 - 1, \alpha)$ if $p_1$ and $n$ have different parities.
\end{itemize}
\end{lemma}

\begin{proof}
The proof is similar to the one for Lemma~\ref{lem:rotatemeander}. 
\end{proof}

In~\cite{LI2009296}, Li and Sawada developed a framework to generate Gray codes for reflectable languages. 
A language $L$ over an alphabet set $\sigma$ is said to be reflectable if for every $i > 1$ there exist two symbols $x_i$ and $y_i$ in $\sigma$ such that if $w_1 w_2 \cdots w_{i-1}$ is a prefix of a word in $L$, then both $w_1 w_2 \cdots w_{i-1}x_i$ and $w_1 w_2 \cdots w_{i-1} y_i$ are also prefixes of words in $L$. 
By reflecting the order of the children and using the special symbols $x_i$ and $y_i$ as the first and last children of each node at level $i - 1$, Li and Sawada devised a generic recursive algorithm to generate Gray codes for reflectable languages which include lots of combinatorial objects such as $k$-ary strings, restricted growth functions and $k$-ary trees. For example, the algorithm generates the binary reflected Gray code in Section~\ref{sec:intro} when we set $x_i =0$ and $y_i =1$ for every $i>1$.

\begin{figure}[t]
\begin{center}
\includegraphics[width=1\columnwidth]{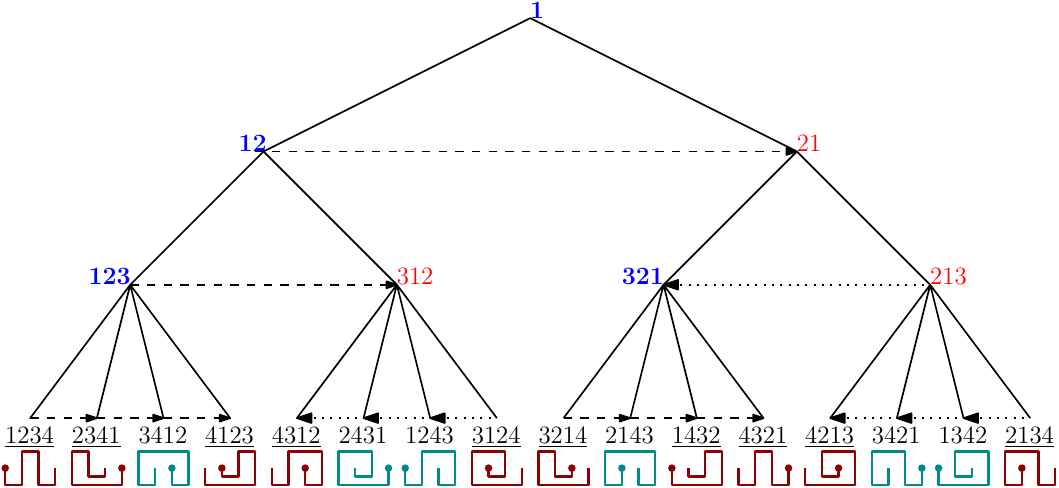}
\end{center}
\caption{
Recursive computation tree constructed by our algorithm {\sc GenR} that outputs stamp foldings and semi-meanders for $n = 4$
in cyclic rotation Gray code order. The nodes at the last level are generated by the PRINT function and
carry no sign. 
For nodes from level $0$ to level $n-2$, the nodes in bold (in blue color) carry a positive sign, and the rest of the nodes
(in red color) carry a negative sign.
A dashed arrow indicates applying right rotations to generate its neighbors, while a dotted arrow indicates applying left rotations to generate its neighbors.
The underlined stamp foldings at the last level of the recursive computation tree are semi-meanders. }
\label{fig:tree_recur}
\end{figure}

\begin{algorithm}[t] 
\caption{An algorithm that finds the number of simple left rotations or simple right rotations required to reach the next semi-meander.}\label{algo:find_steps}
\begin{algorithmic}[1]
\small
\Function{NextSemiMeander}{$p, r, d$}

    \If{$d = 1$} { $j \leftarrow {p_1}$}
    \Else { $j \leftarrow {p_r}$}
    \EndIf

    \If{$j = 1 $ {\bf and} $ r$ is even} {\Return $1$}
    \ElsIf{$j$ and $r$ have the same parity}
        \If{$d = 1$} {\Return $I(j + 1, {p})$}
        \Else { \Return $r - I(j + 1, {p}) + 1$}
        \EndIf
    \Else
        \If{$d = 1$} {\Return $I(j - 1, {p})$}
        \Else { \Return $r - I(j - 1, {p}) + 1$}
        \EndIf
    \EndIf
\EndFunction
\end{algorithmic}
\end{algorithm}
\begin{algorithm}[t] 
\caption{A recursive algorithm that generates rotation Gray codes for stamp foldings and semi-meanders.}\label{algo:recur_gen}
\begin{algorithmic}[1]
\small
\Procedure{GenR}{${p}, t$}
    \State {$i \leftarrow 1$}
    \State {$j \leftarrow 0$}

    \While {$i\leq t+1$}
        \If {$t \geq n-1$} {{\sc Print}$({p})$}
        \Else 
            \If {$q_{t+1} = 1$} {{\sc GenR}$(p\cdot(t+2), t+1)$}
            \Else { {\sc GenR}$((t+2)\cdot p, t+1)$}
            \EndIf
            \State{$q_{t+1} \leftarrow \neg q_{t+1}$}
        \EndIf 
        \If {$t \geq n-1$ {\bf and} generating stamp foldings} {$j \leftarrow 1$}
        \Else { $j \leftarrow ${\sc NextSemiMeander}$({p}, t+1, q_t)$}
        \EndIf
        \If {$q_t = 1$} {$p \leftarrow \overrightarrow{\textrm{rotate}}(j)$ {\color{blue}$\ \ \ \ \ \ \,  \triangleright \  $right rotation}}
        \Else { $p \leftarrow \overleftarrow{\textrm{rotate}}(n-j+1)${\color{blue}$\ \ \ \ \ \ \ \ \ \ \ \ \triangleright \  $left rotation}}
         \EndIf
%        \State {${p} \leftarrow ${\sc Shift}$({p}, j, q_m)$}
        \State {$i \leftarrow i + j$}
    \EndWhile    
\EndProcedure
\end{algorithmic}
\end{algorithm}

Here we use a similar idea to recursively generate a Gray code for semi-meanders. 
The algorithm can be easily modified to generate stamp foldings. 
Depending on the level of a node in the recursive computation tree, there are two possibilities for its children at level $t-1$: 
\begin{itemize}
\item The root of the recursive computation tree is the stamp $1$ of order $1$; 
\item If a node $p_1 p_2 \cdots p_{t-1}$ is at level $t-2$ where $t-2 \geq 0$, then its children are the semi-meanders in $\mathrm{Rots}(t p_1 p_2 \cdots p_{t-1})$, which includes the strings $t p_1 p_2 \cdots p_{t-1}$ and $p_1 p_2 \cdots p_{t-1} t$ (by Corollary~\ref{cor:rotationstamp}).%that are rotational equivalent to $t p_1 p_2 \cdots p_{t-1}$
\end{itemize}
To generate the Gray code for semi-meanders, %we assign a sign to each node of the recursive computation tree.  
we maintain an array $q_0 q_1 \cdots q_{n-1}$ which determines the sign of a new node at each level of the recursive computation tree. 
The current level of the recursive computation tree is given by the parameter $t$.
The sign of a new node at level $t$ is given by the parameter $q_t$, where $q_t = 1$ denotes the new node has a positive sign, and $q_t = 0$ denotes the new node has a negative sign. 
We also initialize the sign $q_t$ at each level as positive ($q_t = 1$) at the beginning of the algorithm.
The word being generated is stored in a doubly linked list $p = p_1 p_2 \cdots p_n$. 
Now if the current node has a positive sign, we generate the child $p_1 p_2 \cdots p_{t-1} t$ and then keep applying right rotations %$\overrightarrow{\textrm{rotate}}_\alpha(i)$ 
to generate all semi-meanders in $\mathrm{Rots}(p_1 p_2 \cdots p_{t-1} t)$ until it reaches the string $t p_1 p_2 \cdots p_{t-1} $.  
Then if the current node has a negative sign,  we generate the child $t p_1 p_2 \cdots p_{t-1} $ and then keep applying left rotations %$\overleftarrow{\textrm{rotate}}_\alpha(j)$ 
to generate all semi-meanders in $\mathrm{Rots}(t p_1 p_2 \cdots p_{t-1})$ until it reaches the string $p_1 p_2 \cdots p_{t-1} t$. 
Finally when we reach level $t = n-1$, we print out the semi-meanders that are in $\mathrm{Rots}(p_1 p_2 \cdots p_{n-1} n)$. 
The function $\textrm{NextSemiMeander}(p, r, d)$ determines the number of simple left rotations or simple right rotations required to find the next semi-meander of $p = p_1 p_2 \cdots p_r$, which is a direct implementation of Lemma~\ref{lem:rotatemeander} and Lemma~\ref{lem:rotatemeander2} (Algorithm~\ref{algo:find_steps}). 
We also complement $q_t$ every time a node is generated at level $t$.  
This way, at the start of each recursive call we can be sure that the previous word generated has stamp $t$ at the same position. 
Finally to generate stamp foldings, notice that removing stamp $n$ from a stamp folding always creates a semi-meander of order $n-1$ (Lemma~\ref{lem:submeander}). The recursive computation tree for stamp foldings is thus the same as the one that generates semi-meanders, except at level $n-1$ we print all stamp foldings that are in $\mathrm{Rots}(p_1 p_2 \cdots p_{n-1} n)$, that is all strings in $\mathrm{Rots}(p_1 p_2 \cdots p_{n-1} n)$. 
Pseudocode of the algorithm is shown in Algorithm~\ref{algo:recur_gen}. 
To run the algorithm, we make the initial call $\textrm{GenR}(1,0)$ which sets $p = 1$ and  $t = 0$.

As an example, the algorithm generates the following cyclic rotation Gray codes for stamp foldings and semi-meanders of length five respectively:
\begin{itemize}
\label{eqn:somelabel}
\item {
$\overrightarrow{\underline{1}{2345}}, \overrightarrow{\underline{2}{3451}},  
\overrightarrow{\underline{3}{4512}}, \overrightarrow{\underline{4}{5123}},
5\overrightarrow{\underline{1}{234}}, \overleftarrow{5234\underline{1}},  
\overleftarrow{1523\underline{4}}, \overleftarrow{4152\underline{3}}, 
\overleftarrow{3415\underline{2}}, \overrightarrow{\underline{23}{41}}5, 
\overrightarrow{\underline{4}{1235}}, \overrightarrow{\underline{1}{2354}},$ \\
$\overrightarrow{\underline{2}{3541}}, \overrightarrow{\underline{3}{5412}},  
54\overrightarrow{\underline{12}{3}}, \overleftarrow{5431\underline{2}},
\overleftarrow{2543\underline{1}}, \overleftarrow{1254\underline{3}},  
\overleftarrow{3125\underline{4}}, \overleftarrow{{4}\underline{312}}5, 
\overrightarrow{\underline{3}{1245}}, \overrightarrow{\underline{1}{2453}}, 
\overrightarrow{\underline{2}{4531}}, \overrightarrow{\underline{4}{5312}},$ \\
$53\overrightarrow{\underline{1}{2}}4, \overleftarrow{5321\underline{4}},  
\overleftarrow{4532\underline{1}}, \overleftarrow{1453\underline{2}},  
\overleftarrow{2145\underline{3}}, \overrightarrow{\underline{32}{14}}5,  
\overrightarrow{\underline1{4325}}, \overrightarrow{\underline{4}{3251}}, 
\overrightarrow{\underline{3}{2514}}, \overrightarrow{\underline{2}{5143}},  
5\overrightarrow{\underline{1}{432}}, \overleftarrow{5432\underline{1}},$ \\ 
$\overleftarrow{1543\underline{2}}, \overleftarrow{2154\underline{3}}, 
\overleftarrow{3215\underline{4}}, 4\overleftarrow{{3}\underline{21}}5, 
\overrightarrow{\underline{4}{2135}}, \overrightarrow{\underline{2}{1354}},  
\overrightarrow{\underline{1}{3542}}, \overrightarrow{\underline{3}{5421}},  
5\overleftarrow{{4}\underline{213}}, \overleftarrow{5213\underline{4}}, 
\overleftarrow{4521\underline{3}}, \overleftarrow{3452\underline{1}}, $\\ 
$\overleftarrow{1345\underline{2}},  \overleftarrow{2\underline{1}}345; $}

\item {
$\overrightarrow{\underline{12}{345}}, \overrightarrow{\underline{34}{512}}, 5\overrightarrow{\underline{1}{234}}, \overleftarrow{5\underline{2341}}, \overrightarrow{\underline{23}{41}}5, \overrightarrow{\underline{4123}{5}},  
54\overrightarrow{\underline{12}{3}}, \overleftarrow{543\underline{12}}, 
\overleftarrow{125\underline{43}}, \overleftarrow{4\underline{312}}5, 
\overrightarrow{\underline{3124}{5}}, 53\overrightarrow{\underline{1}{2}}4,$ \\ 
$\overleftarrow{5\underline{3214}}, \overrightarrow{\underline{32}{14}}5, 
\overrightarrow{\underline{1432}{5}}, 5\overrightarrow{\underline{1}{432}},  
\overleftarrow{543\underline{21}}, \overleftarrow{215\underline{43}},  
4\overleftarrow{3\underline{21}}5, \overrightarrow{\underline{4213}{5}},
5\overleftarrow{4\underline{213}}, \overleftarrow{521\underline{34}}, 
\overleftarrow{345\underline{21}}, \overleftarrow{2\underline{1}}345.$} 

\end{itemize}
The Gray code listing of semi-meanders can also be obtained by \emph{filtering} the Gray code listing of stamp foldings. 
For more about filtering Gray codes, see~\cite{10.1007/978-3-030-85088-3_15, SAWADA2022138}.  
Figure~\ref{fig:tree_recur} illustrates the recursive computation tree when $n = 4$.

\section{Analyzing the recursive algorithm}
\label{sec:proves}
%Define $LR^k_{i, j}(\alpha)$ and $RR^k_{i, j}(\alpha)$.
In this section, we prove that our recursive algorithm generates cyclic rotation Gray codes for stamp foldings and semi-meanders in constant amortized time and $O(n)$-amortized time per string respectively. 

We first prove the Gray code property for semi-meanders. 
The proof for the Gray code property for stamp foldings follows from the one for semi-meanders.

\begin{lemma} \label{lem:meander-gray-semi}
Each consecutive semi-meanders in the listing generated by the algorithm {\sc GenR} differ by a stamp rotation ${\textrm{rotate}}_\alpha(i, j, k)$ for some $k < i \leq j \leq n$. 
%The algorithm {\sc Gen} produces a list of all stamp foldings of length $n$ in rotation Gray code order, where consecutive strings differ by a rotation ${\textrm{rotate}}_\alpha(i, j, k)$ for some $k \leq i \leq j$.
\end{lemma}
\begin{proof}
The proof is by induction on $n$.
In the base case when $n = 2$, the generated semi-meanders are $12$ and $21$ and clearly they differ by a stamp rotation. 
Inductively, assume consecutive semi-meanders generated by the algorithm {\sc GenR} differ by a stamp rotation when $n = k-1$. 
Consider the case when $n = k$, clearly the nodes in levels $0$ to $k-2$ of the recursive computation tree for $n = k$ are exactly the same as the recursive computation tree for $n = k-1$. 
%Now consider the nodes at level $k$. 
If two consecutive nodes at level $k-1$ of the recursive computation tree have the same parent at level $k-2$, then clearly the semi-meanders differ by a left rotation or a right rotation. 
Otherwise if two consecutive nodes at level $k-1$ of the recursive computation tree have different parents $\alpha$ and $\beta$ at level $k-2$.
Observe that by the algorithm  {\sc GenR}, $\alpha$ and $\beta$ are consecutive nodes at level $k-2$ of the recursive computation tree and W.L.O.G. assume $\alpha$ comes before $\beta$. 
Then by the inductive hypothesis, the corresponding semi-meanders for $\alpha$ and $\beta$ differ by a stamp rotation. 
Also, since $\alpha$ and $\beta$ are consecutive nodes at level $k-2$ of the recursive computation tree, their signs are different by the algorithm and thus the two children can only be the last child of $\alpha$ and the first child of $\beta$. 
If $\alpha$ carries a positive sign, then the two children are $n \cdot \alpha$ and $n \cdot \beta$. Otherwise, the two children are $\alpha \cdot n$ and $\beta \cdot n$. 
In both cases, the two children differ by a stamp rotation. 
\end{proof}

\begin{lemma} \label{lem:meander-circular}
The first and last strings generated by the algorithm {\sc GenR} are $123 \cdots n$ and $2134 \cdots n$ respectively.
\end{lemma}

\begin{proof}
The leftmost branch of the recursive computation tree corresponds to the first string generated by the algorithm. 
Observe that all nodes of the leftmost branch carry positive signs and are the first child of their parents. 
Thus the first string generated by the algorithm is $123 \cdots n$. 
Similarly, the rightmost branch of the recursive computation tree corresponds to the last string generated by the algorithm.
The rightmost node at level one of the recursive computation tree corresponds to the string $21$. 
Furthermore, observe that the number of semi-meanders and the number of stamp foldings are even numbers when $n > 1$ (for each stamp folding (semi-meander) $p_1 p_2 \cdots p_n$, its reversal $p_n p_{n-1} \cdots p_1$ is also a stamp folding (semi-meander)). 
Thus the rightmost nodes at each level $t>0$ carry negative signs and are the last child of their parents. 
Therefore, the last string generated by the algorithm is $2134 \cdots n$. 
\end{proof}

\begin{lemma} \label{lem:meander-exhau}
The algorithm {\sc GenR} exhaustively generates all semi-meanders of length $n$.
\end{lemma}

\begin{proof}
The proof is by induction on $n$. 
In the base case when $n=2$, the generated semi-meanders are $12$ and $21$ which are all the possible semi-meanders for $n=2$. 
Inductively, assume the algorithm generates all semi-meanders for $n =k-1$. 
Consider the case when $n = k$, by Lemma~\ref{lem:submeander} clearly the algorithm generates either one of the semi-meanders of the form $\{k p_1 p_2 \cdots p_{k-1}, p_1 p_2 \cdots p_{k-1} k\}$ for all possible $p_1 p_2 \cdots p_{k-1}$ when considering the first child produced by each node at level $k-2$. 
Since each semi-meander has a symbol $p_t = k$, generating all valid string rotations of the first child  exhaustively lists out all semi-meanders for $n = k$. 
\end{proof}

Together, Lemma~\ref{lem:meander-gray-semi}, Lemma~\ref{lem:meander-circular} and Lemma~\ref{lem:meander-exhau} prove the following theorem.
\begin{theorem} \label{thm:meander-main}
The algorithm {\sc GenR} produces a list of all semi-meanders of order $n$ in a cyclic rotation Gray code order.
\end{theorem}

We then prove the Gray code property for stamp foldings.

\begin{lemma} \label{lem:meander-gray-stamp}
Each consecutive stamp foldings in the listing generated by the algorithm {\sc GenR} differ by a stamp rotation ${\textrm{rotate}}_\alpha(i, j, k)$ for some $k < i \leq j \leq n$. 
\end{lemma}

\begin{proof}
Since levels $0$ to $n - 2$ of the recursive computation tree
%the first $n-1$ levels of the recursive computation tree 
for stamp foldings of length $n$ are exactly the same as the recursive computation tree for semi-meanders of length $n-1$, by Lemma~\ref{lem:meander-gray-semi} the strings correspond to consecutive nodes at level $n-2$ of the recursive computation tree differ by a stamp rotation. 
Therefore by applying the same argument as the inductive step of Lemma~\ref{lem:meander-gray-semi}, consecutive stamp foldings at level $n-1$ of the recursive computation tree also differ by a stamp rotation. 
\end{proof}

\begin{lemma} \label{lem:stamps-exhau}
The algorithm {\sc GenR} exhaustively generates all stamp foldings of length $n$.
\end{lemma}

\begin{proof}
Since levels $0$ to $n - 2$ of the recursive computation tree for stamp foldings of length $n$ are exactly the same as the recursive computation tree for semi-meanders of length $n-1$, by Lemma~\ref{lem:meander-exhau} the algorithm exhaustively generates all nodes correspond to semi-meanders of length $n-1$. 
Thus by Lemma~\ref{lem:submeander}, the algorithm generates either one of the stamp foldings of the form $\{n p_1 p_2 \cdots p_{n-1}, p_1 p_2 \cdots p_{n-1} n\}$ for all possible length $n-1$ semi-meanders $p_1 p_2 \cdots p_{n-1}$  when considering the first child produced by each node at level $n-2$.  
Since the set of stamp foldings is partitioned into equivalence classes under string rotation, the algorithm generates all string rotations of the first child and thus exhaustively lists out all stamp foldings.
\end{proof}

Similarly, Lemma~\ref{lem:meander-gray-stamp}, Lemma~\ref{lem:meander-circular} and Lemma~\ref{lem:stamps-exhau} prove the following theorem.
\begin{theorem} \label{thm:stamp-main}
The algorithm {\sc GenR} produces a list of all stamp foldings of order $n$ in a cyclic rotation Gray code order.
\end{theorem}

Finally, we prove the time complexity of our algorithm. 

\begin{theorem} \label{thm:semi-time}
The algorithm {\sc GenR} generates semi-meanders and stamp foldings of length $n$ in cyclic rotation Gray code order in $O(n)$-amortized time and constant amortized time per string respectively, using $O(n)$ space.
\end{theorem}
\begin{proof}
Clearly for each node at level $t<n-1$, each recursive call of the algorithm {\sc GenR} only requires $O(n)$ amount of work and a linear amount of space to generate all rotations of a semi-meander. 
By Corollary~\ref{cor:rotationstamp}, since each call to {\sc GenR} makes at least two recursive calls and there are no dead ends in the computation tree, the algorithm generates each node at level $n-2$ of the computation tree in $O(n)$-amortized time per node using a linear amount of space. 
If we are generating semi-meanders, by the same argument each node at level $n-1$ requires $O(n)$ amount of work and thus each string can be generated in $O(n)$-amortized 
time per string using a linear amount of space. 
Otherwise if we are generating stamp foldings, by Lemma~\ref{lem:rotation} each node at level $n-2$ of the computation tree has exactly $n$ children, while as discussed above each node at level $n-2$ can be generated in $O(n)$-amortized time per node. 
Therefore, the algorithm generates stamp foldings in constant amortized time per string using a linear amount of space.
\end{proof}

\section{An iterative algorithm to generate Gray codes for semi-meanders
and stamp foldings}
\label{sec:iterative}
%%% Iterative Method
In this section, we present a simple iterative algorithm for generating the same cyclic rotation Gray codes for stamp foldings and semi-meanders as the recursive algorithm described in Section~\ref{sec:algorithm}. 
%The algorithm takes advantage of the observation that the maximum rotation steps for a meander at level $t$ is $t - 1$. 
%By maintaining a length $n$ binary array $q$, we can keep track of the rotation directions and the information about completely rotated child meanders.
\begin{algorithm}[t] 
\small
\caption{An iterative algorithm that generates rotation Gray codes for stamp
foldings and semi-meanders.}\label{algo:iter_meander}
\begin{algorithmic}[1]
\Procedure{GenI}{$n$}
    
    \State {$p_1 p_2 \cdots p_{n} \gets 1 \cdot 2 \cdots n$}
    \State {$q_0 q_1 \cdots q_{n-1} \leftarrow 1^n$}
%    \State {$t \gets n-1$}

%    \While{$t > 0$}
     \Do
        \State {{\sc Print}$({p})$}
        \State {$i \gets 1$}
%        \State {$l \gets -1$}
        \For {$t$ from $n-1$ to $0$}
            \If{$q_t = 1$ and $p_i = t+1$}
                \State{$q_t \gets \neg q_t$}
                \State{$i \gets i + 1$}
            \ElsIf {$q_t = 0$ and $p_{i+t} = t+1$}
                \State {$q_t \gets \neg q_t$}
            \Else { \textbf{break}}
            \EndIf
        \EndFor
        
%        \If{$l = -1$} {\textbf{exit}}
%        \EndIf

        \If{$t = n-1$}
            \If {generating stamp foldings} 
            \State {$j \gets 1$}
            \For {$i$ from $0$ to $n-2$}
            \If {$q_t = 1$} {$p \leftarrow \overrightarrow{\textrm{rotate}}(1)$}
            \Else { $p \leftarrow \overleftarrow{\textrm{rotate}}(n)$}
            \EndIf
            \State {{\sc Print}$({p})$}
            \EndFor
        \Else { $j \gets ${\sc NextSemiMeander}$(p, t+1, q_t)$}
                \EndIf

            \If {$q_t = 1$} {$p \leftarrow \overrightarrow{\textrm{rotate}}(j)$}
            \Else { $p \leftarrow \overleftarrow{\textrm{rotate}}(n - j + 1)$}
            \EndIf
        \ElsIf{$t>0$}
        \State {$j \gets ${\sc NextSemiMeander}$(p_i p_{i+1} \cdots p_{i+t}, t+1, q_t)$}
        \If{$q_t$ = 1} {$p \gets p_1 p_2 \cdots p_{i-1} \cdot \overrightarrow{\textrm{rotate}}_{p_i p_{i+1} \cdots p_{i+t}}(j) \cdot p_{i+t+1} p_{i+t+2} \cdots p_n$}
        \Else { $p \gets p_1 p_2 \cdots p_{i-1} \cdot \overleftarrow{\textrm{rotate}}_{p_i p_{i+1} \cdots p_{i+t}}(n - j + 1) \cdot p_{i+t+1} p_{i+t+2} \cdots p_n$}
        \EndIf
        \EndIf
    \doWhile{$t > 0$}
\EndProcedure
\end{algorithmic}
\end{algorithm}

The iterative algorithm for generating cyclic rotation Gray codes for stamp foldings and semi-meanders operates by maintaining a binary array $q = q_0 q_1 \ldots q_{n-1}$ of length $n$, where each element $q_t$ represents the rotation direction at level $t$: 1 for right rotation and 0 for left rotation. 
Initially, $q$ is set to all ones, indicating right rotations at each level. 
The first string, $p$, is initialized as the sequence $p = 1\cdot2 \cdots n$. 
Starting at the highest level $t = n-1$, the iterative algorithm performs rotations according to the direction indicated by $q_t$ until reaching the maximum position, determined by the leftmost character being $t+1$ (when $q_t = 1$) or the rightmost character being $t+1$ (when $q_t = 0$). 
If $t = n-1$, the next stamp folding can be found by a simple left rotation or a simple right rotation, depending on $q_t$. 
Otherwise, the next semi-meander can be found using the function $\textrm{NextSemiMeander}(p, t+1, q_t)$. 
If a rotation occurs at the maximum position at level $t$, $t$ is decremented to find the next level $k$ that is not rotated at the maximum position, and the corresponding rotation according to $q_k$ is performed. 
The positions in $q$ from $q_{k+1}$ up to $q_{n-1}$ are then reversed, and the process repeats to find the largest possible level $t$ again for rotation. 
This continues until $t = 0$, at which point the algorithm stops. 
We also maintain a variable $i$ that keeps track of the substring $p_i p_{i+1} \cdots p_{i+t}$ of $p$ that corresponds to the semi-meander of level $t$. 
This is based on Corollary~\ref{cor:rotationstamp}, which implies that the largest character in a semi-meander must occupy either the first or last position when it reaches its maximum rotation.
By following this iterative approach, the algorithm generates the same cyclic rotation Gray codes for stamp foldings and semi-meanders as the recursive algorithm.
Pseudocode of the iterative algorithm is given in Algorithm~\ref{algo:iter_meander}.
To run the algorithm, we make the initial call $\textrm{GenI}(n)$. 

As an example, the algorithm generates the same cyclic rotation Gray codes
for stamp foldings and semi-meanders as {\sc GenR} as shown in Section~\ref{sec:algorithm}. 

\begin{theorem} \label{thm:meander-iterative-main}
The algorithm {\sc GenI} produces the same cyclic rotation Gray codes generated by the algorithm {\sc GenR} for semi-meanders and stamp foldings of order $n$.
\end{theorem}

\begin{proof}
Both the {\sc GenI} and {\sc GenR} algorithms begin with the same initial string $p = 1\cdot2 \cdots n$. 
In the recursive computation tree of the algorithm {\sc GenR}, each level initially rotates in the right direction, which corresponds to the binary array $q = q_0 q_1 \cdots q_{n-1} = 1^n$. 
The recursive algorithm {\sc GenR} then starts rotating the stamp folding or semi-meander from the maximum depth of the recursive computation tree. 
This corresponds to the process of the iterative algorithm {\sc GenI} to find the maximum level $t$ that has not reached its maximum rotation.

To handle the reflection process in the recursive algorithm {\sc GenR}, the direction of $q_{k+1}$ to $q_{n-1}$ is reversed once a maximal level $k$ is found that has not reached its maximal rotation. 
As such, the iterative algorithm {\sc GenI} generates an equivalent listing to that of the recursive algorithm {\sc GenR}.
\end{proof}
We would also like to note that this iterative approach can be extended to the generation of reflectable Gray codes as discussed in~\cite{LI2009296} by using a similar methodology.

\label{sec:iterative-prove}

Finally, we prove the time complexity of our iterative algorithm.

\begin{theorem} \label{thm:semi-time-interative}
The algorithm {\sc GenI} generates semi-meanders and stamp foldings of length $n$ in cyclic rotation Gray code order in $O(n)$ time and constant amortized time per string respectively, using $O(n)$ space.
\end{theorem}

\begin{proof}
Each left rotation $\overleftarrow{\textrm{rotate}}_\alpha(i)$ and right rotation $\overrightarrow{\textrm{rotate}}_\alpha(j)$ can be completed in constant time by using a doubly linked list to maintain each stamp folding or semi-meander. 
Also, the function $\textrm{NextSemiMeander}(p, t+1, q_t)$  can be called to find the next semi-meander under left rotation or right rotation, which takes $O(n)$ time. 
Moreover, the operations required to maintain and update the binary array $q$ can be easily accomplished in $O(n)$ time. It is thus easy to see that the algorithm {\sc GenI}  generates semi-meanders in $O(n)$ time per string, using $O(n)$ space.

Finally for stamp folding, observe that the function $\textrm{NextSemiMeander}(p, t+1, q_t)$ is called only when $t<n-1$, which is also the case when $n$ reaches the leftmost position when $q_{n-1} = 1$ or the rightmost position when $q_{n-1} = 0$. %, which is also the case when we need to change the direction for $q_n$. 
The for loop from line 7 to 13 in Algorithm~\ref{algo:iter_meander} would then complement $q_{n-1}$ in the next do-while iteration, which causes $t = n-1$ and generates $n-1$ stamp foldings in the for loop from line 17 to 20 in Algorithm~\ref{algo:iter_meander}. 
As such, the runtime for the algorithm {\sc GenI} to generate each stamp folding is amortized to constant amortized time per string.
\end{proof}

It is worth pointing out that {\sc GenI} algorithm offers the generation of semi-meanders in $O(n)$ time per string, instead of $O(n)$-amortized time per string by {\sc GenR}. 
This can be considered as a slight advantage when comparing the two algorithms.

\section{Final remarks}
\label{sec:remarks}
This paper presents a recursive algorithm and an iterative algorithm for generating rotation Gray codes for stamp foldings and semi-meanders. 
Both algorithms generate stamp foldings and semi-meanders in constant amortized time per string and $O(n)$-amortized time per string respectively, using a linear amount of space.

An \emph{open meander} is a stamp folding with the restriction that stamp 1 can be seen from above, and stamp $n$ can be seen
from above when $n$ is even and stamp $n$ can be seen from below when $n$ is odd.
For example, the permutations that correspond to the six open meanders in Figure~\ref{fig:stamp}  for $n=4$ are given as follows: 
\begin{center}
%1234, 1243, 1342, 1432, 2341, 2431, 3214, 3421, 4123, 4321. 
1234, 1432, 2341, 3214, 4123, 4321. 
\end{center}
The problem of finding Gray codes for open meanders is listed as an open problem in the paper by Sawada and Li~\cite{DBLP:journals/combinatorics/SawadaL12}, and also in the Gray code survey by M\"{u}tze~\cite{mutze2023combinatorial}. 
However, even using a similar strategy mentioned in this paper, we were unable to produce a Gray code for open meanders.  
We leave this as an open problem. 

\section*{Acknowledgements}
\noindent
This research is supported by the Macao Polytechnic University research grant (Project code: RP/FCA-02/2022) and the National Research Foundation of Korea (NRF) grant funded by the Ministry of Science and ICT (MSIT), Korea (No. 2020R1F1A1A01070666).

\newpage
\bibliographystyle{abbrv}
\bibliography{myrefs}       % expects file "myrefs.bib"

\clearpage

\noindent
\large
{\bf Appendix A: C code of recursive algorithm to generate stamp foldings and semi-meanders in cyclic rotation Gray code order}

\tiny
\lstset{style=mystyle}
\begin{lstlisting}[language=C]
#include <stdio.h>
#include <stdlib.h>
#define MAX 100

int n, q[MAX], total = 0, type;

struct linked_list {
    int element;
    struct linked_list *prev;
    struct linked_list *next;
};

//--------------------------------------------------
struct linked_list* Rotate(struct linked_list *p, int k, int d) 
{
    if (d)
        for (int i = 0; i < k%n; i++) p = p->prev;
    else
        for (int i = 0; i < k%n; i++) p = p->next;
    
    return p;
}

//--------------------------------------------------
int IndexOf(int e, struct linked_list *p, int r, int t) {
    if (r) {
        for (int i = 1; i <= n; i++) {
            if (p->element == e) return i;
            p = p->next;
        }
    }
    else {
        p = p->prev;
        for (int i = 1; i <= n; i++) {
            if (p->element == e) return t-i+1;
            p = p->prev;
        }
    }
    return 1;
}

//--------------------------------------------------
int NextSemiMeander(struct linked_list *p, int r, int d) {
    int j;

    if (d) j = p->element;
    else j = p->prev->element;

    if (j == 1 && !(r%2)) return 1;

    else if (j%2 == r%2) {
        if (d) return IndexOf(j+1, p, d, r);
        else return r - IndexOf(j+1, p, d, r) + 1;
    }
    else {
        if (d) return IndexOf(j-1, p, d, r);
        else return r - IndexOf(j-1, p, d, r) + 1;
    }
}

//--------------------------------------------------
void Print(struct linked_list *p) {
    for (int i = 0; i < n; i++) {
        printf("%d ", p->element);
        p = p->next;
    }
    printf("\n");   total += 1;
}

//--------------------------------------------------
void GenR(struct linked_list *p, int t) {
    int i = 1, j = 0;
    
    while (i <= t+1) {
        if (t >= n-1) Print(p);
        else {
            struct linked_list *r = NULL;
            r = (struct linked_list*)
            malloc(sizeof(struct linked_list));
            r->element = t+2;
            r->next = p;    r->prev = p->prev;
            r->next->prev = r;  r->prev->next = r;

            if (q[t+1]) GenR(p, t+1);
            else GenR(r, t+1);
            
            q[t+1] = !q[t+1];
            p->prev = r->prev;  p->prev->next = p;
        }
        
        if (type == 1 && t >= n-1) j = 1; 
        else j = NextSemiMeander(p, t+1, q[t]);
        p = Rotate(p, j, !q[t]);
        i = i + j;
    }
}

//--------------------------------------------------
int main() {
    printf("ENTER n: "); scanf("%d", &n);
    printf("Enter selection # (1. Stamp foldings, 2. Semi-meanders): ");
    scanf("%d", &type);

    for (int i = 0; i <= n; i++) q[i] = 1;
    
    struct linked_list *p = NULL;
    p = (struct linked_list*)malloc(sizeof(struct linked_list));
    p->element = 1;
    p->prev = p;    p->next = p;

    GenR(p, 0);
    printf("\nTotal: %d\n", total);
}
\end{lstlisting}

\clearpage
\noindent
\large
{\bf Appendix B: C code of iterative algorithm to generate stamp foldings and semi-meanders in cyclic rotation Gray code order}

\tiny
\lstset{style=mystyle}
\begin{lstlisting}[language=C]
#include <stdio.h>
#include <stdlib.h>
#define MAX 100

int n, q[MAX], type;
long total = 0;

struct linked_list {
    int element;
    struct linked_list *prev;
    struct linked_list *next;
};

struct linked_list *Rotate(struct linked_list *p, int k, int d) 
{
    if (d) for (int i = 0; i < k%n; i++) p = p->next;
    else for (int i = 0; i < k%n; i++) p = p->prev;
    return p;
}

//--------------------------------------------------
int IndexOf(int e, struct linked_list *p, int r, int t) {
    if (r) {
        for (int i = 1; i <= n; i++) {
            if (p->element == e) return i;
            p = p->next;
        }
    }
    else {
        p = p->prev;
        for (int i = 1; i <= n; i++) {
            if (p->element == e) return t-i+1;
            p = p->prev;
        }
    }
    return 1;
}

//--------------------------------------------------
int NextSemiMeander(struct linked_list *p, int r, int d) {
    int j;

    if (d) j = p->element;
    else j = p->prev->element;

    if (j == 1 && !(r%2)) return 1;

    else if (j%2 == r%2) {
        if (d) return IndexOf(j+1, p, d, r);
        else return r - IndexOf(j+1, p, d, r) + 1;
    }
    else {
        if (d) return IndexOf(j-1, p, d, r);
        else return r - IndexOf(j-1, p, d, r) + 1;
    }
}

//--------------------------------------------------
void Print(struct linked_list *p) {
    for (int i = 0; i < n; i++) {
        printf("%d ", p->element);
        p = p->next;
    }
    printf("\n");
    total += 1;
}

//--------------------------------------------------
void GenI(struct linked_list* p) {
    int t, j;
    struct linked_list *head, *tail, *beforeHead, *afterTail;

    do {
        Print(p);
        head = p;   tail = p->prev;
        
        for (t = n-1; t >= 0; t--) {
            if (q[t] && head->element == t+1) {
                q[t] = 1 - q[t];
                head = head->next; tail = tail->next;
            }
            else if (!q[t] && tail->element == t+1) 
                q[t] = 1 - q[t];
            else break;
            tail = tail->prev;
        }

        if (t == n-1) {
            if (type == 1) {
                j = 1;
                // produce n stamp foldings if t != n-1
                for (int i = 0; i < n-2; i++) { 
                    p = Rotate(p, 1, q[t]);
                    Print(p);
                    head = p;   tail = p->prev;
                }
            }
            else j = NextSemiMeander(p, t+1, q[t]);
            p = Rotate(p, j, q[t]);
        }
        else if (t > 0) {
            beforeHead = head->prev; afterTail = tail->next;
            head->prev = tail; tail->next = head;
            j = NextSemiMeander(head, t+1, q[t]);
            if (p == head) {
                head = Rotate(head, j, q[t]);
                p = head;
            }
            else head = Rotate(head, j, q[t]);
            beforeHead->next = head; 
            afterTail->prev = head->prev;
            beforeHead->next->prev = beforeHead; 
            afterTail->prev->next = afterTail;
        }
    } while (t > 0);
}

//--------------------------------------------------
int main() {
    printf("ENTER n: ");
    scanf("%d", &n);
    printf("Enter selection # (1. Stamp foldings, 2. Semi-meanders): ");
    scanf("%d", &type);

    for (int i = 0; i <= n; i++) q[i] = 1;

    struct linked_list *p = NULL;
    p =(struct linked_list *)malloc(sizeof(struct linked_list));
    p->element = 1;

    struct linked_list *temp = p;
    for (int i = 2; i <= n; i++) {
        temp->next = (struct linked_list *) 
        malloc(sizeof(struct linked_list));
        temp->next->element = i;
        temp->next->prev = temp;
        temp = temp->next;
    }
    p->prev = temp; temp->next = p;

    GenI(p);
    printf("\nTotal: %ld \n", total);
}
\end{lstlisting}

\end{document}